\documentclass[british]{scrartcl}
\pdfoutput=1
\usepackage[T1]{fontenc}
\usepackage[utf8]{inputenc}
\usepackage{lmodern}
\usepackage{amsthm}
\usepackage{amsmath}
\usepackage{amssymb}
\usepackage{xfrac}
\usepackage[hidelinks]{hyperref}

\theoremstyle{plain}
\newtheorem{prop}{Proposition}
\theoremstyle{definition}
\newtheorem{defn}[prop]{Definition}

\title{A note on the expected minimum error probability in equientropic channels}

\author{Sebastian Weichwald, Tatiana Fomina,\\
    Bernhard Schölkopf, Moritz Grosse-Wentrup\\
    {\small Max Planck Institute for Intelligent Systems, Tübingen, Germany}\\
    {\small\texttt{[sweichwald, tfomina, bs, moritzgw]@tue.mpg.de}}}

\date{}

\begin{document}

\maketitle

\section{Introduction}

While the channel capacity reflects a theoretical upper bound on the achievable information transmission rate in the limit of infinitely many bits, it does not characterise the information transfer of a given encoding routine with finitely many bits.
In this note, we characterise the quality of a code (i.\,e. a given encoding routine) by an upper bound on the expected minimum error probability that can be achieved when using this code.
We show that for equientropic channels this upper bound is minimal for codes with maximal marginal entropy.
As an instructive example we show for the additive white Gaussian noise (AWGN) channel that random coding---also a capacity achieving code---indeed maximises the marginal entropy in the limit of infinite messages.

\section{Upper bounding the expected minimum error probability}

Consider communication over \emph{noisy memoryless channels}
\[
M\xrightarrow{\text{encoder}}X^{n}\xrightarrow{\text{channel}}Y^{n}\xrightarrow{\text{decoder}}\widehat{M}
\]
where the sender node $M$ is a random variable taking discrete values $m\in\mathcal{M}$ according to $p_{M}$;
the values $x^{n}=[x_{1},...,x_{n}]$ of the sender bits $X^{n}$ are determined by the encoder function $f_{\text{enc}}:\mathcal{M}\to\mathcal{X}^{n}$ assigning codewords to messages;
noise corruption of the received bits $Y^{n}$ is governed by the conditional distribution $p_{Y|X}$ as $p_{Y^{n}|X^{n}}\left(y^{n}|x^{n}\right)=\prod_{j=1}^{n}p_{Y|X}\left(y_{j}|x_{j}\right)$;%
\footnote{To ease notation we assume $p_{Y|X}=p_{Y_{j}|X_{j}}$ for all $j\in\mathbb{N}_{1:n}$. The results presented in this manuscript only require a memoryless channel and still hold true if noise corruption is bit-specific, i.\,e., $p_{Y^{n}|X^{n}}=\prod_{j=1}^{n}p_{Y_{j}|X_{j}}$.}
and the decoder $f_{\text{dec}}:\mathcal{Y}^{n}\to\mathcal{M}\cup\{e\}$ reconstructs a message from the received bit values or declares an error.
The message distribution $p_{M}$, the encoder $f_{\text{enc}}$, the channel $p_{Y|X}$, and the decoder $f_{\text{dec}}$ fully determine the distribution of the receiver node $\widehat{M}$ and as such the probability of error $\mathbb{P}\left[M\neq\widehat{M}\right]$.

Thence, for given message distribution $p_{M}$ and channel $p_{Y|X}$ the \emph{code}, that is the choice of $f_{\text{enc}}$ (and corresponding $f_{\text{dec}}$), fully determines the behaviour of information
transmission.
The minimum probability of error is attained if choosing the maximum a posteriori (MAP) decoder $\arg\max_{m\in\mathcal{M}}p_{M|Y^{n}}\left(m|y^{n}\right)$.
Thus, for any code $f_{\text{enc}}$, the expected minimum error probability is the MAP error $\mathcal{E}\left(f_{\text{enc}}\right):=\mathbb{E}_{Y^{n}}\left[1-\max_{m\in\mathcal{M}}p_{M|Y^{n}}\left(m|y^{n}\right)\right]$.
We characterise the quality of a code $f_{\text{enc}}$ by the following Proposition.

\begin{prop}\label{prop:map-bound}
For communication of a message $M\sim p_{M}$ with finite range over a noisy memoryless channel $p_{Y|X}$ using $n$ bits the MAP error $\mathcal{E}\left(f_{\text{enc}}\right)$ can be bounded in terms of the mutual information $I\left(Y^n;M\right) = H\left(Y^{n}\right) - H\left(Y^{n}|M\right)$ as
\[
\gamma\left(-I\left(Y^{n};M\right)\right)\leq\mathcal{E}\left(f_{\text{enc}}\right)\leq\Gamma\left(-I\left(Y^{n};M\right)\right)
\]
where $\gamma$ and $\Gamma$ are strictly monotonically increasing functions.
\end{prop}

\begin{proof}
\cite[Theorem 1]{Feder1994} establishes the following relation (notation adapted)
\[
\Phi\left(\mathcal{E}\left(f_{\text{enc}}\right)\right)\geq H\left(M|Y^{n}\right)\geq\phi^{*}\left(\mathcal{E}\left(f_{\text{enc}}\right)\right)
\]
where $\Phi$ and $\phi^{*}$ are continuous and strictly monotonically increasing, hence invertible, functions (cf.~\cite{Feder1994} for their definitions).
Recall $H\left(M|Y^{n}\right)=H\left(M\right)+H\left(Y^{n}|M\right)-H\left(Y^{n}\right)$ and note that $H\left(M\right)$ is fix for fixed $p_{M}$.
The inequality follows for $\gamma\left(h\right):=\Phi^{-1}\left(H\left(M\right)+h\right)$ and $\Gamma(h):=\phi^{*^{-1}}\left(H\left(M\right)+h\right)$ which are strictly monotonically increasing functions in $h$.
\end{proof}

That is, codes $f_{\text{enc}}$ that result in high $I\left(Y^n;M\right) = H\left(Y^{n}\right) - H\left(Y^{n}|M\right)$ result in a low upper bound on the MAP error.
In particular, of all codes resulting in the same conditional entropy $H\left(Y^{n}|M\right)$ a code with maximal entropy $H\left(Y^{n}\right)$ has the lowest upper bound on the MAP error.
The following Propositions simplify this result for equientropic channels and independent additive noise channels: The lowest upper bound on the MAP error is achieved for codes $f_\text{enc}$ that maximise the entropy of receiver bits $H\left(Y^n\right)$ and the entropy of sender bits $H\left(X^n\right)$, respectively.

\begin{defn}
A noisy memoryless channel $p_{Y|X}$ with $H\left(Y|X=x_{1}\right)=H\left(Y|X=x_{2}\right)$ for all $x_{1},x_{2}\in\mathcal{X}$ is an \emph{equientropic channel}.
\end{defn}

\begin{prop}
For equientropic channels $p_{Y|X}$ the conditional entropy $H\left(Y^{n}|M\right)$ is independent of the choice of $f_{\text{enc}}$.
\end{prop}

\begin{proof}
The channel is memoryless such that $H\left(Y^{n}|M\right)=\sum_{j=1}^{n}H\left(Y_{j}|M\right)$.
For any $x\in\mathcal{X}$ and $j\in\mathbb{N}_{1:n}$
\[
H\left(Y_{j}|M\right) = \sum_{m\in\mathcal{M}} p_{M}\left(m\right)H\left(Y_{j}|X_{j}=f_{\text{enc}}\left(m_{i}\right)_{j}\right)=\sum_{m\in\mathcal{M}} p_{M}\left(m_{i}\right)H\left(Y_{j}|X_{j}=x\right)
\]
which shows that $H\left(Y_{j}|M\right)$ and hence $H\left(Y^{n}|M\right)$ is independent of the choice of $f_{\text{enc}}$.
\end{proof}

\begin{defn}
A noisy memoryless channel $p_{Y|X}$ with $Y^n|X^n = X^n + N^n$ for mutually independent noise variables $N^n \sim p_{N^n} = \prod_{i=1}^n p_{N_i}$ that are independent of $X^n$ is an \emph{independent additive noise channel}.
Independent additive noise channels are equientropic channels.
\end{defn}

\begin{prop}
For independent additive noise channels with noise variables $N^n$ the entropy of the receiver bits $H\left(Y^n\right) = H\left(X^n\right) + H\left(N^n\right)$ only depends on the choice of $f_\text{enc}$ via the entropy of the sender bits $H\left(X^n\right)$.
\end{prop}

In conclusion, optimality of a code $f_\text{enc}$ for communication over a noisy memoryless channel with message distribution $M\sim p_M$ can be characterised by the upper bound on the MAP error that results from this code.
The respective bounds for different channels are summarised in Table~\ref{tab:optimal}.
Importantly, without knowing specific details about the channel and decoder, maximising entropy turns out to be a sensible heuristic for learning a robust coding routine.
Intuitively, high entropy distributed codes are more robust against independent noise.

\begin{table}[t]
\caption{Upper bounds on the MAP error $\mathcal{E}\left(f_\text{enc}\right)$ for communication of a message $M\sim p_{M}$ with finite range over different channels where $\Gamma, \Gamma', \Gamma''$ are strictly monotonically increasing functions.}
\label{tab:optimal}
\centering
\renewcommand{\arraystretch}{1.2}
\begin{tabular}{l|l}
Channel Type & Bound \\
\hline
noisy memoryless channel & $\mathcal{E}\left(f_\text{enc}\right) \leq \Gamma\left(-I\left(Y^n;M\right)\right)$ \\
equientropic channel & $\mathcal{E}\left(f_\text{enc}\right) \leq \Gamma'\left(-H\left(Y^n\right)\right)$ \\
independent additive noise channel & $\mathcal{E}\left(f_\text{enc}\right) \leq \Gamma''\left(-H\left(X^n\right)\right)$
\end{tabular}
\end{table}

\section{AWGN random coding example}

The AWGN channel is an ubiquitous and well-understood channel model.
Here it serves as an instructive example for the concept introduced in the previous section.

The AWGN channel is an independent additive noise channel and described by
\begin{alignat*}{1}
Z     & \sim \mathcal{N}\left(0,NI_{n\times n}\right)\\
Y_{i} & = gX_{i}+Z_{i}\text{ for }i\in\mathbb{N}_{1:n}
\end{alignat*}
where $g$ is the channel gain and $N$ the noise level.
We employ the power constraint that each codeword $x^{n}=f_{\text{enc}}(m)\in\mathcal{X}^{n}$ has to satisfy
\[
\frac{1}{n}\sum_{i=1}^{n}\left(x_{i}\right)^{2}\leq P
\]
and without loss of generality assume $N=1$ such that the received power is $S=g^{2}P$.
The Shannon-Hartley theorem establishes the channel capacity \[C=\max_{p_{X}:\mathbb{E}_{X}\left[X^{2}\right]\leq P}I(X;Y)=\frac{1}{2}\log\left(1+S\right)\]
Achievability of this upper bound on the rate is commonly proven by random coding, i.\,e., for any rate $R:=\frac{\log_{2}\left|\mathcal{M}\right|}{n}\leq C$ the error probability tends to zero as $n=\log_{2}\left|\mathcal{M}\right|\to\infty$ if using random coding.

Here we show that random coding not only achieves the optimal rate but also the lowest upper bound on the MAP error in Proposition~\ref{prop:map-bound} since $H\left(Y^{n}\right)=\sum_{i=1}^{n}H\left(Y_{i}\right)$ (and
the $Y_{i}$ are Gaussian maximising the individual entropies) in the limit $n\to\infty$.

In random coding the encoder function $f_{\text{r-enc}}$ is defined by a random codebook, i.\,e., an independent sample of $C^{n}\sim\mathcal{N}\left(0,PI_{n\times n}\right)$ is assigned to each message $m_{i}$ as codeword $f_{\text{r-enc}}\left(m_{i}\right)=[c_{i1},...,c_{in}]$.
Once a codebook is fixed and we observe samples of the system each receiver bit $Y_{j}$ is a mixture of Gaussians with probability densitiy function (pdf) $p_{Y_{j}}\left(y_{j}\right)=\sum_{i=1}^{\left|\mathcal{M}\right|}p_{M}\left(m_{i}\right)\varphi\left(y_{j}|c_{ij},1\right)$ where $\varphi\left(y|\mu,\sigma^{2}\right)$ denotes the pdf of the Gaussian distribution $\mathcal{N}\left(\mu,\sigma^{2}\right)$ evaluated at $y$.
For this setup we prove the following

\begin{prop}
Using random coding in the AWGN channel with $p_{M}\sim\operatorname{Unif}\left(\mathcal{M}\right)$ the joint entropy $H\left(Y_{j_{1}},...,Y_{j_{k}}\right)\xrightarrow[n\to\infty]{\text{almost surely}}\sum_{l=1}^{k}H\left(Y_{j_{l}}\right)$ for any number of $k$ pairwise different receiver bits $Y_{j_{1}},...,Y_{j_{k}}$.
Furthermore, the distribution of each $Y_{j}$ approaches a Gaussian distribution $\mathcal{N}\left(0,P+1\right)$ as $n\to\infty$.
\end{prop}

\begin{proof}
In random coding the random codebook is generated by drawing each $c_{ij_{l}}$ from independent random variables $C_{ij_{l}}\sim\mathcal{N}\left(0,P\right)$, which then defines the joint pdf
\[
p_{Y_{j_{1}},...,Y_{j_{k}}}\left(y_{j_{1}},...,y_{j_{k}}\right)=\sum_{i=1}^{\left|\mathcal{M}\right|}p_{M}\left(m_{i}\right)\left(2\pi\right)^{-\frac{k}{2}}e^{-\frac{1}{2}\sum_{l=1}^{k}\left(y_{j_{l}}-c_{ij_{l}}\right)^{2}}
\]
and marginal pdfs
\[
p_{Y_{j_{l}}}\left(y_{j_{l}}\right)=\sum_{i=1}^{\left|\mathcal{M}\right|}p_{M}\left(m_{i}\right)\left(2\pi\right)^{-\frac{1}{2}}e^{-\frac{1}{2}\left(y_{j_{l}}-c_{ij_{l}}\right)^{2}}
\]
for $l\in\mathbb{N}_{1:k}$ and $y_{j_1}, ..., y_{j_k} \in \mathcal{Y}$.
In general $p_{Y_{j_{1}},...,Y_{j_{k}}}\neq\prod_{l=1}^{k}p_{Y_{j_{l}}}$.

For all $l\in\mathbb{N}_{1:k}$ and $y_{j_1}, ..., y_{j_k} \in \mathcal{Y}$ define the random variables
\[
\mathring{p}_{Y_{j_{1}},...,Y_{j_{k}}}\left(y_{j_{1}},...,y_{j_{k}}\right)=\frac{1}{\left|\mathcal{M}\right|}\sum_{i=1}^{\left|\mathcal{M}\right|}\left(2\pi\right)^{-\frac{k}{2}}e^{-\frac{1}{2}\sum_{l=1}^{k}\left(y_{j_{l}}-C_{ij_{l}}\right)^{2}}
\]
and
\[
\mathring{p}_{Y_{j_{l}}}\left(y_{j_{l}}\right)=\frac{1}{\left|\mathcal{M}\right|}\sum_{i=1}^{\left|\mathcal{M}\right|}\left(2\pi\right)^{-\frac{1}{2}}e^{-\frac{1}{2}\left(y_{j_{l}}-C_{ij_{l}}\right)^{2}}
\]

By the law of large numbers
\begin{align*}
\mathring{p}_{Y_{j_{1}},...,Y_{j_{k}}}\left(y_{j_{1}},...,y_{j_{k}}\right)
    & \xrightarrow[n\to\infty]{\text{almost surely}}
    \mathbb{E}_{C_{1j_{1}},...,C_{1j_{k}}} \left[
        \left(2\pi\right)^{-\frac{k}{2}}
        e^{-\frac{1}{2}\sum_{l=1}^{k}\left(y_{j_{l}}-C_{1j_{l}}\right)^{2}}
    \right]\\
\mathring{p}_{Y_{j_{l}}}\left(y_{j_{l}}\right)
    & \xrightarrow[n\to\infty]{\text{almost surely}}
    \mathbb{E}_{C_{1j_{l}}} \left[
        \left(2\pi\right)^{-\frac{1}{2}}
        e^{-\frac{1}{2}\left(y_{j_{l}}-C_{1j_{l}}\right)^{2}}
    \right]
\end{align*}
where the first expectation factorises since the $C_{1j_{1}},...,C_{1j_{k}}$ are mutually independent.
It follows that for all $y_{j_{1}},...,y_{j_{k}}\in\mathcal{Y}$
\[
\mathring{p}_{Y_{j_{1}},...,Y_{j_{k}}}\left(y_{j_{1}},...,y_{j_{k}}\right)-\prod_{l=1}^{k}\mathring{p}_{Y_{j_{l}}}\left(y_{j_{l}}\right)\xrightarrow[n\to\infty]{\text{almost surely}}0
\]
such that in the limit the pdf indeed factorises.
Evaluating the expectation above we find that for each $Y_{j}$ and $y_{j}$ $\mathring{p}_{Y_{j}}\left(y_{j}\right)\xrightarrow[n\to\infty]{\text{almost surely}}\left(2\pi\left(P+1\right)\right)^{-\frac{1}{2}}e^{-\frac{1}{2(P+1)}y_{j}^{2}}=\varphi\left(y_{j}|0,P+1\right)$ which concludes the proof.
\end{proof}

It is instructive to consider the analogous statement for any $k$ pairwise different sender bits $X_{j_1}, ..., X_{j_k}$.
The proof follows analogous arguments and is another illustration of the fact that in independent additive noise channels the bound on the MAP error is fully determined by the entropy of the sender bits $H\left(X^n\right) = H\left(Y^n\right) - H\left(Z^n\right)$.

\section{Further thoughts}

According to the efficient coding hypothesis the brain implements an efficient code for representing sensory input by neuronal spiking \cite{barlow1961}.
Observed dependencies between neurons and hence redundancies are sometimes viewed as contradicting the efficient coding hypothesis \cite{barlow1961,simoncelli2003}.
The results presented in Section~2 clarify, however, that an optimal code should maximise the joint entropy $H\left(Y^n\right)$ of receiver (or sender) bits.
For fixed marginal entropies $H\left(Y_j\right)$ the maximum is indeed achieved if all units are mutually independent.
However, since the marginal entropies are not fixed there can in general be configurations that have higher joint entropy while the units are not mutually independent.
This also clarifies the intuition expressed in Shannon's early work that the transmitted signals should approximate white noise to approximate the maximum information rate \cite[Section~25.]{Shannon1948}.

\bibliographystyle{alpha}
\bibliography{references}

\end{document}